\newcommand{\lv}[1]{#1}
\newcommand{\sv}[1]{}

\documentclass[a4paper,USenglish]{lipics}

\usepackage{times}

\usepackage{tikz}
\usetikzlibrary{positioning,calc}

\usepackage{url}
\usepackage{amsmath,amssymb}
\usepackage{xspace,tikz,enumitem}
\usepackage{mathrsfs,verbatim,hyperref}
\usepackage{enumitem,todonotes,microtype}

\usepackage{color}
\usepackage{boxedminipage}

\usepackage[]{algorithm2e}

\newcommand{\Nat}{\mathbb{N}}
\newcommand{\hy}{\hbox{-}\nobreak\hskip0pt}

\newcommand{\bigoh}{\mathcal{O}}

\def\hy{\hbox{-}\nobreak\hskip0pt} 

\newcommand{\SB}{\{\,} \newcommand{\SM}{\;{|}\;} \newcommand{\SE}{\,\}}
 
\newcommand{\Card}[1]{|#1|}

\newcommand{\CC}{\mathcal{C}} 
 
\newcommand{\FF}{\mathcal{F}} 
\newcommand{\WW}{\mathcal{W}} 
\newcommand{\RR}{\mathcal{R}} 
\newcommand{\WWW}{\mathcal{W}} 

\newcommand{\NP}{\text{\normalfont NP}}

\newcommand{\mtext}[1]{\text{\normalfont\itshape #1}} 

\newcommand{\gleq}{\mathop{\diamondsuit}}

\def\mx#1{\mbox{\boldmath$#1$}}

\def\MS#1{\mbox{MSO}}

\newcommand{\MSO}{\mbox{MSO}\xspace}

\def\EE{{\mathcal E}}

\def\HH{{\mathcal H}}
\def\KK{{\mathcal K}}

\def\PP{{\mathcal P}}
\def\ZZ{{\mathcal Z}}

\newtheorem{proposition}{Proposition}
\newtheorem{fact}{Fact}

{\bfseries}{\itshape}
\newtheorem{observation}{Observation}{\bfseries}{\itshape}

\def\sm{split-module}

\def\wsm{well-structured modulator}

\newcommand{\MSOMC}[1]{\textsc{MSO-MC${}_{#1}$}}
\newcommand{\MSOOPT}[2]{\textsc{MSO-Opt${}^{#1}_{#2}$}}

\newcommand{\wsn}{\mtext{wsn}}
\newcommand{\rwc}{\mtext{rwc}}
\newcommand{\vcn}{\mtext{vcn}}

\newcommand{\rw}{\mtext{rw}}
\newcommand{\md}{\mtext{mod}}
\newcommand{\hit}{\mtext{hit}}

\title{Meta-Kernelization using Well-Structured Modulators\footnote{Supported by the Austrian Science Fund (FWF), project P26696.}}

\author{Eduard Eiben}
\author{Robert Ganian}
\author{Stefan Szeider}

\affil{Algorithms and Complexity Group, TU Wien\\
  Vienna, Austria}
\authorrunning{E. Eiben and R. Ganian and S. Szeider}

\Copyright{Eduard Eiben and Robert Ganian and Stefan Szeider}


\begin{document}

\maketitle

\begin{abstract}
  \noindent Kernelization investigates exact preprocessing algorithms
  with performance guarantees.  The most prevalent type of parameters
  used in kernelization is the solution size for optimization
  problems; however, also structural parameters have been successfully
  used to obtain polynomial kernels for a wide range of problems. Many
  of these parameters can be defined as the size of a smallest
  \emph{modulator} of the given graph into a fixed graph class (i.e.,
  a set of vertices whose deletion puts the graph into the graph
  class). Such parameters admit the construction of polynomial kernels
  even when the solution size is large or not applicable. This work
  follows up on the research on meta-kernelization frameworks in terms
  of structural parameters.

  We develop a class of parameters which are based on a more general
  view on modulators: instead of size, the parameters employ a
  combination of rank-width and split decompositions to measure
  structure inside the modulator. This allows us to lift kernelization
  results from modulator-size to more general parameters, hence
  providing smaller kernels. We show (i)~how such large but
  well-structured modulators can be efficiently approximated, (ii)~how
  they can be used to obtain polynomial kernels for any graph problem
  expressible in Monadic Second Order logic, and (iii)~how they allow
  the extension of previous results in the area of structural
  meta-kernelization.
\end{abstract}


\section{Introduction}


Kernelization investigates exact preprocessing algorithms with
performance guarantees. Similarly as in parameterized complexity
analysis, in kernelization we study \emph{parameterized problems}:
decision problems where each instance $I$ comes with a parameter
$k$. A parameterized problem is said to admit a kernel of size
$f:\Nat \rightarrow \Nat$ if every instance $(I,k)$ can be reduced in
polynomial time to an equivalent instance (called the \emph{kernel})
whose size and parameter are bounded by $f(k)$. For practical as well
as theoretical reasons, we are mainly interested in the existence of
\emph{polynomial kernels}, i.e., kernels whose size is polynomial in
$k$. The study of kernelization has recently been one of the main
areas of research in parameterized complexity, yielding many important
new contributions to the theory. 

The by far most prevalent type of parameter used in kernelization is the \emph{solution size}. Indeed, the existence of polynomial kernels and the exact bounds on their sizes have been studied for a plethora of distinct problems under this parameter, and the rate of advancement achieved in this direction over the past $10$ years has been staggering. Important findings were also obtained in the area of \emph{meta-kernelization}~\cite{BodlaenderEtal09,FominLokshtanovSaurabh10,KimLangerPaulReidlRossmanith13}, which is the study of general kernelization techniques and frameworks used to establish polynomial kernels for a wide range of distinct problems. 

In parameterized complexity analysis, an alternative to
parameterization by solution size has traditionally been the use of
\emph{structural parameters}. But while parameters such as
\emph{treewidth} and the more general \emph{rank-width} allow the
design of FPT algorithms for a range of important problems, it is
known that they cannot be used to obtain polynomial kernels for
problems of interest. Instead, the structural parameters used for
kernelization often take the form of the size of minimum
\emph{modulators} (a modulator of a graph is a set of vertices whose
deletion puts the graph into a fixed graph class). Examples of such
parameters include the size of a minimum vertex
cover~\cite{FominJansenP14,BodlaenderJansenKratsch13} (modulators into
the class of edgeless graphs) or of a minimum feedback vertex
set~\cite{BodlaenderJansenKratsch13b,JansenBodlaender13} (modulators
into the class of forests). While such structural parameters are not
as universal as the structural parameters used in the context of
fixed-parameter tractability, these results nonetheless allow
efficient preprocessing of instances where the solution size is large
and for problems where solution size simply cannot be used (such as
$3$-coloring).

This paper follows up on the recent line of research which studies meta-kernelization in terms of structural parameters. Gajarsk\' y et al.~\cite{GajarskyHlinenyObdrzalek13} developed a meta-kernelization framework parameterized by the size of a modulator to the class of graphs of bounded treedepth on sparse graphs. Ganian et al.~\cite{GanianSlivovskySzeider13} independently developed a meta-kernelization framework using a different parameter based on rank-width and modular decompositions (see Subsection~\ref{sub:parcomp} for details). Our results build upon both of the aforementioned papers by fully subsuming the meta-kernelization framework of~\cite{GanianSlivovskySzeider13} and lifting the meta-kernelization framework of~\cite{GajarskyHlinenyObdrzalek13} to more general graph classes. The class of problems investigated in this paper are problems which can be expressed using \emph{Monadic Second Order} (MSO) logic (see Subsection~\ref{sub:mso}).

The parameters for our kernelization results are also based on
modulators. However, instead of parameterizing by the \emph{size} of
the modulator, we instead measure the \emph{structure} of the
modulator through a combination of rank-width and split
decompositions. Due to its technical nature, we postpone the
definition of our parameter, the \emph{well-structure number}, to
Section~\ref{sec:kcwsm}; for now, let us roughly describe it as the
number of sets one can partition a modulator into so that each set
induces a graph with bounded rank-width and a simple neighborhood. We
call modulators which satisfy our conditions \emph{well-structured}. A
less restricted variant of the well-structure number has recently been
used to obtain meta-theorems for FPT algorithms on graphs of unbounded
rank-width~\cite{EibenGanianSzeider15}.

After formally introducing the parameter, in Section~\ref{sec:vc} we showcase its applications on the special case of generalizing the \emph{vertex cover number} by considering \wsm s to edgeless graphs. While it is known that there exist MSO-definable problems which do not admit a polynomial kernel parameterized by the vertex cover number on general graphs, on graphs of bounded expansion this is no longer the case (as follows for instance from~\cite{GajarskyHlinenyObdrzalek13}). On the class of graphs of bounded expansion, we prove that every \MSO{}-definable problem admits a linear kernel parameterized by the well-structure number for edgeless graphs. As a corollary of our approach, we also show that every \MSO{}-definable problem admits a linear kernel parameterized by the well-structure number for the empty graph (without any restriction on the expansion). We remark that the latter result represents a direct generalization of the results in~\cite{GanianSlivovskySzeider13}. The proof is based on a combination of a refined version of the replacement techniques developed in~\cite{EibenGanianSzeider15} together with the annotation framework used in~\cite{GanianSlivovskySzeider13}.

Before we can proceed to wider applications of our parameter in kernelization, it is first necessary to deal with the subproblem of finding a suitable \wsm{} in polynomial time. We resolve this question for well-structured modulators to a vast range of graph classes. In particular, in Subsection~\ref{sub:fvs} we obtain a $3$-approximation algorithm for finding \wsm s to acyclic graphs, and in the subsequent Subsection~\ref{sub:fis} we show how to approximate \wsm s to any graph class characterized by a finite set of forbidden induced subgraphs within a constant factor.

Section~\ref{sec:using} then contains our most general result,
Theorem~\ref{thm:main-use}, which is the key for lifting kernelization
results from modulators to \wsm s. The theorem states that whenever a
modulator to a graph class $\HH$ can be used to poly-kernelize some
\MSO{}-definable problem, this problem also admits a polynomial kernel
when parameterized by the well-structure number for $\HH$ as long as
\wsm s to $\HH$ can be approximated in polynomial time. The remainder
of Section~\ref{sec:using} then deals with the applications of this
theorem. Since the class of graphs of treedepth bounded by some fixed
integer can be characterized by a finite set of forbidden induced
subgraphs, we can use \wsm s to lift the results
of~\cite{GajarskyHlinenyObdrzalek13} from modulators to \wsm s for all
\MSO{}-definable decision problems. Furthermore, by applying the
\emph{protrusion} machinery
of~\cite{BodlaenderEtal09,KimLangerPaulReidlRossmanith13} we show
that, in the case of bounded degree graphs, parameterization by a
modulator to acyclic graphs (i.e., a feedback vertex set) allows the
computation of a linear kernel for any \MSO{}-definable decision
problem. By our framework it then follows that such modulators can
also be lifted to \wsm s.

\sv{\smallskip \noindent {\emph{Statements whose proofs are located in the appendix are marked with $\star$.}}}

\section{Preliminaries}\label{sec:prel}
The set of natural numbers (that is, positive integers) will be denoted by
$\Nat$. For $i \in \Nat$ we write $[i]$ to denote the set $\{1,
\dots, i \}$. If $\sim$ is an equivalence relation over a set $A$, then for $a\in A$ we use $[a]_{\sim}$ to denote the equivalence class containing $a$.

\subsection{Graphs} We will use standard graph theoretic terminology and notation
(cf. \cite{Diestel00}). All graphs in this document are simple and undirected. 

Given a graph $G=(V(G),E(G))$ and $A\subseteq V(G)$, we denote by $N(A)$ the set of neighbors of $A$ in $V(G)\setminus A$; if $A$ contains a single vertex $v$, we use $N(v)$ instead of $N(\{v\})$. We use $V$ and $E$ as shorthand for $V(G)$ and $E(G)$, respectively, when the graph is clear from context. $G-A$ denotes the subgraph of $G$ obtained by deleting $A$.
For $A\subseteq V(G)$ we use $G[A]$ to denote the subgraph of $G$ obtained by deleting $V(G)\setminus A$.

\subsection{Splits and Split-Modules}
  A \emph{split} of a connected graph $G=(V,E)$ is a vertex bipartition $\{A,B\}$ of $V$
 such that every vertex of $A' = N(B)$ has the same neighborhood in $B'=N(A)$. The sets $A'$ and $B'$ are called \emph{frontiers} of the split. 
 \newcommand{\prelimsplitsa}[0]{
 A split is said to be \emph{non-trivial} if both sides have at least two vertices. A connected graph which does not contain a non-trivial split is called \emph{prime}. A bipartition is \emph{trivial} if one of its parts is the empty set or a singleton. Cliques and stars are called \emph{degenerate} graphs; notice that every non-trivial bipartition of their vertices is a split.}
 \lv{\prelimsplitsa}

Let $G=(V,E)$ be a graph. To simplify our exposition, we will use the notion of \emph{\sm s} instead of splits where suitable. A set $A\subseteq V$ is called a \emph{\sm} of $G$ if there exists a connected component $G'=(V',E')$ of $G$ such that $\{A,V'\setminus A\}$ forms a split of $G'$. Notice that if $A$ is a {\sm} then $A$ can be partitioned into $A_1$ and $A_2$ such that $N(A_2)\subseteq A$ and for each $v_1, v_2\in A_1$ it holds that $N(v_1)\cap (V'\setminus A)=N(v_2)\cap (V'\setminus A)$; $A_1$ is then called the frontier of $A$. For technical reasons, $V$ and $\emptyset$ are also considered \sm s. We say that two disjoint \sm s~$X, Y \subseteq V$ are \emph{adjacent} if there exist
$x \in X$ and $y \in Y$ such that $x$ and $y$ are adjacent. We use $\lambda(A)$ to denote the frontier of \sm{} $A$.

\subsection{Rank-Width}

For a graph $G$ and $U,W\subseteq V(G)$, let $\mx A_G[U,W]$ denote the
$U\times W$-submatrix of the adjacency matrix over the two-element
field $\mathrm{GF}(2)$, i.e., the entry $a_{u,w}$, $u\in U$ and $w\in
W$, of $\mx A_G[U,W]$ is $1$ if and only if $\{u,w\}$ is an edge
of~$G$.  The {\em cut-rank} function $\rho_G$ of a graph $G$ is
defined as follows: For a bipartition $(U,W)$ of the vertex
set~$V(G)$, $\rho_G(U)=\rho_G(W)$ equals the rank of $\mx A_G[U,W]$
over $\mathrm{GF}(2)$. 

A \emph{rank-decomposition} of a graph $G$ is a pair $(T,\mu)$
where $T$ is a tree of maximum degree 3 and $\mu:V(G)\rightarrow \{t:
\text{$t$ is a leaf of $T$}\}$ is a bijective function. For an edge~$e$ of~$T$, the connected components of $T - e$ induce a
bipartition $(X,Y)$ of the set of leaves of~$T$.  The \emph{width} of
an edge $e$ of a rank-decomposition $(T,\mu)$ is $\rho_G(\mu^{-1} (X))$.
The \emph{width} of $(T,\mu)$ is the maximum width over all edges of~$T$.  The \emph{rank-width} of $G$, $\rw(G)$ in short, is the minimum width over all
rank-decompositions of $G$. We denote by $\RR_i$ the class of all graphs of rank-width at most $i$, and say that a graph class $\HH$ is of \emph{unbounded rank-width} if $\HH\not \subseteq \RR_i$ for any $i\in \Nat$.

\newenvironment{psmallmatrix}
  {\left(\begin{smallmatrix}}
  {\end{smallmatrix}\right)}

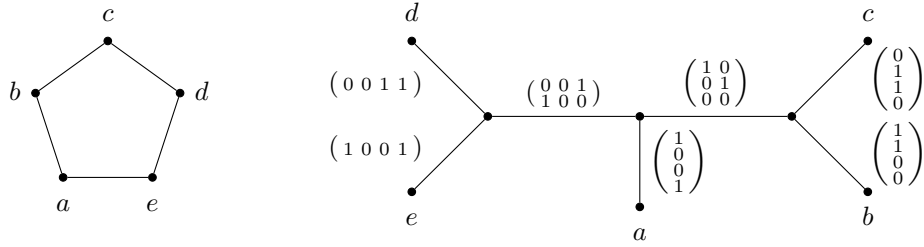
\begin{figure}[ht]
\vspace{-0.3cm}
  \centering
  \tikzstyle{every circle node}=[circle,draw,inner sep=1.0pt, fill=black]
  \begin{tikzpicture}
    \begin{scope}[xshift=-7cm]
      \draw
      (90+4*72:1) node[circle] (d)  {} node[right=2pt] {$\strut d$} 
      (90+0*72:1) node[circle] (c)  {} node[above] {$\strut c$} 
      (90+1*72:1) node[circle] (b)  {} node[left=2pt] {$\strut b$} 
      (90+2*72:1) node[circle] (a)  {} node[below] {$\strut a$} 
      (90+3*72:1) node[circle] (e)  {} node[below] {$\strut e$} 
      (a)--(b)--(c)--(d)--(e)--(a)
      ;
    \end{scope}

 \draw 
 (-3,1)   node[circle] (d) {} node[above] {$\strut d$} 
 (-3,-1)  node[circle] (e) {} node[below] {$\strut e$} 
 (-2,0)   node[circle] (U) {} 
 (0,0)    node[circle] (V) {} 
 (0,-1.2)    node[circle] (X) {} node[below] {$\strut a$} 
 (2,0)    node[circle] (W) {}
 (3,1)   node[circle] (c) {} node[above] {$\strut c$} 
 (3,-1)  node[circle] (b) {} node[below] {$\strut b$} 

 (d) to node [auto,near start, swap] {$\begin{psmallmatrix}
                              0&0&1&1
                            \end{psmallmatrix}$} (U)
 (e) to node [auto, near start] {$\begin{psmallmatrix}
                              1&0&0&1
                            \end{psmallmatrix}$} (U)
 (U) to node [auto] {$\begin{psmallmatrix}
                              0&0&1\\
                              1&0&0  
                            \end{psmallmatrix}$} (V)
 (V) to node [auto] {$\begin{psmallmatrix}
                              1&0\\
                              0&1\\
                              0&0
                            \end{psmallmatrix}$} (W)
 (V)  to node [auto] {$\begin{psmallmatrix}
                              1\\
                              0\\
                              0\\
                              1
                            \end{psmallmatrix}$} (X)
 (c)--(W)--(b)
(3.4,0.5) node {$\begin{psmallmatrix}
                              0\\
                              1\\
                              1\\
                              0
                            \end{psmallmatrix}$} 
(3.4,-0.5) node {$\begin{psmallmatrix}
                              1\\
                              1\\
                              0\\
                              0
                            \end{psmallmatrix}$} 
;
\end{tikzpicture}
\vspace{-0.5cm}
\caption{A rank-decomposition of the cycle $C_5$.}
\label{fig:rdecC5}
\end{figure}

\begin{fact}[\cite{HlinenyOum08}]\label{fact:rankdecomp} Let $k \in \Nat$ be a constant and
  $n \geq 2$. For an $n$-vertex graph $G$, we can output a
  rank-decomposition of width at most $k$ or confirm that the
  rank-width of $G$ is larger than $k$ in time $\bigoh(n^3)$.
\end{fact}

More properties of rank-width can be found, for instance, in~\cite{OumSeymour06}.

\subsection{Fixed-Parameter Tractability and Kernels}
\label{sub:parcomp}

\newcommand{\QQ}{Q}

A \emph{parameterized problem} $\PP$ is a subset of $\Sigma^* \times
\Nat$ for some finite alphabet $\Sigma$. For a problem instance $(x,k)
\in \Sigma^* \times \Nat$ we call $x$ the main part and $k$ the
parameter.  A parameterized problem $\PP$ is \emph{fixed-parameter
  tractable} (FPT in short) if a given instance $(x, k)$ can be solved in time
$O(f(k) \cdot p(\Card{x}))$ where $f$ is an arbitrary computable
function of $k$ and $p$ is a polynomial function.

A \emph{bikernelization} for a parameterized problem $\PP \subseteq
\Sigma^* \times \Nat$ into a parameterized problem $\QQ \subseteq
\Sigma^* \times \Nat$ is an algorithm that, given $(x, k) \in \Sigma^*
\times \Nat$, outputs in time polynomial in $\Card{x}+k$ a pair $(x',
k') \in \Sigma^* \times \Nat$ such that (i)~$(x,k) \in \PP$ if and
only if $(x',k') \in \QQ$ and (ii)~$\Card{x'}+k'\leq g(k)$, where $g$
is an arbitrary computable function.  The reduced instance $(x',k')$
is the \emph{bikernel}.  If $\PP=\QQ$, the reduction is called a
\emph{kernelization} and $(x',k')$ a \emph{kernel}.  The function $g$
is called the \emph{size} of the (bi)kernel, and if $g$ is a
polynomial then we say that $\PP$ admits a \emph{polynomial
  (bi)kernel}.
  
The following fact links the existence of bikernels to the existence of kernels.
  
\begin{fact}[\cite{AlonGutinKimSY11}]
\label{fact:bikernels}
Let $\PP,\QQ$ be a pair of decidable parameterized problems such that $\QQ$ is in \NP{} and $\PP$ is \NP-complete. If there is a bikernelization from $\PP$ to $\QQ$ producing a polynomial bikernel, then $\PP$ has a polynomial kernel.
\end{fact}

Within this paper, we will also consider (and compare to) various structural parameters which have been used to obtain polynomial kernels. We provide a brief overview of these parameters below.

A \emph{modulator} of a graph $G$ to a graph class $\HH$ is a vertex set $X\subseteq V(G)$ such that $G-X\in \HH$. We denote the cardinality of a minimum modulator to $\HH$ in $G$ by $\md^\HH(G)$.
The \emph{vertex cover number} of a graph $G$ ($\vcn(G)$) is a special case of $\md^\HH(G)$, specifically for $\HH$ being the set of edgeless graphs. The vertex cover number has been used to obtain polynomial kernels for problems such as \textsc{Largest Induced Subgraph}~\cite{FominJansenP14} or \textsc{Long Cycle} along with other path and cycle problems~\cite{BodlaenderJansenKratsch13}. 
Similarly, a \emph{feedback vertex set} is a modulator to the class of acyclic graphs, and the size of a minimum feedback vertex set has been used to kernelize, for instance, \textsc{Treewidth}~\cite{BodlaenderJansenKratsch13b} or \textsc{Vertex Cover}~\cite{JansenBodlaender13}.

For the final considered parameter, we will need the notion of \emph{module}, which can be defined as a \sm{} with the restriction that every vertex in the \sm{} lies in its frontier. Then the \emph{rank-width$_c$ cover number}~\cite{GanianSlivovskySzeider13} of a graph $G$ ($\rwc_c(G)$) is the smallest number of \emph{modules} the vertex set of G can be partitioned into such that each module induces a subgraph of rank-width at most $c$. A wide range of problems, and in particular all $\MSO$-definable problems, have been shown admit linear kernels when parameterized by the rank-width$_c$ cover number~\cite{GanianSlivovskySzeider13}.


\subsection{Monadic Second Order Logic on Graphs}
\label{sub:mso}
We assume that we have an infinite supply of individual variables,
denoted by lowercase letters $x,y,z$, and an infinite supply of set
variables, denoted by uppercase letters $X,Y,Z$. \emph{Formulas} of
\emph{monadic second-order logic} (MSO) are constructed from atomic
formulas $E(x,y)$, $X(x)$, and $x = y$ using the connectives $\neg$
(negation), $\wedge$ (conjunction) and existential quantification
$\exists x$ over individual variables as well as existential
quantification $\exists X$ over set variables. Individual variables
range over vertices, and set variables range over sets of
vertices. The atomic formula $E(x,y)$ expresses adjacency, $x = y$
expresses equality, and $X(x)$ expresses that vertex $x$ in the set
$X$. From this, we define the semantics of monadic second-order logic
in the standard way (this logic is sometimes called $\MSO_1$).

\emph{Free and bound variables} of a formula are defined in the usual way. A
\emph{sentence} is a formula without free variables. We write $\varphi(X_1,
\dots, X_n)$ to indicate that the set of free variables of formula $\varphi$
is $\{X_1, \dots, X_n\}$. If $G = (V,E)$ is a graph and $S_1, \dots, S_n
\subseteq V$ we write $G \models \varphi(S_1, \dots, S_n)$ to denote that
$\varphi$ holds in $G$ if the variables $X_i$ are interpreted by the sets
$S_i$, for $i \in [n]$. The problem framework we are mainly interested in is formalized below.

\begin{quote}
  \textsc{MSO Model Checking} ($\MSOMC{\varphi}$)\\
  \nopagebreak \emph{Instance}: A graph $G$. \\ \nopagebreak 
  \emph{Question}: Does $G \models \varphi$ hold?
\end{quote}


While MSO model checking problems already capture many important graph
problems, there are some well-known problems on graphs
that cannot be captured in this way, such as \textsc{Vertex
  Cover}, \textsc{Dominating Set}, and \textsc{Clique}. Many such problems can be formulated in the form of \emph{MSO optimization problems}.  Let $\varphi=\varphi(X)$ be an MSO
formula with one free set variable~$X$ and $\gleq \in \{\leq,\geq\}$.

\begin{quote}
  $\MSOOPT{\gleq}{\varphi}$\\
  \nopagebreak \emph{Instance}: A graph $G$ and an integer~$r\in\Nat$. \\ \nopagebreak 
  \emph{Question}: Is there a set $S \subseteq V(G)$ such that $G \models \varphi(S)$ and $\Card{S} \gleq r$?
\end{quote}

It is known that MSO formulas can be checked efficiently as long as the graph has bounded rank-width.

\begin{fact}[\cite{GanianHlineny10}]\label{fact:msorankwidth}
  Let $\varphi$ and $\psi=\psi(X)$ be fixed \MSO formulas and let $c$ be a constant. Then $\MSOMC{\varphi}$ and $\MSOOPT{\gleq}{\varphi}$ can be solved in $\bigoh(n^3)$ time on the class of graphs of rank-width at most $c$, where $n$ is the order of the input graph. Moreover, if $G$ has rank-width at most $c$ and and $S\subseteq V(G)$, it is possible to check whether $G\models \psi(S)$ in $\bigoh(n^3)$ time.
\end{fact}

We review \MSO \emph{types} roughly following the presentation in
\cite{Libkin04}. The \emph{quantifier rank} of an \MSO formula $\varphi$ is
defined as the nesting depth of quantifiers in $\varphi$. For non\hy negative
integers $q$ and $l$, let $\MSO_{q,l}$ consist of all \MSO formulas of
quantifier rank at most $q$ with free set variables in $\{X_1, \dots, X_l\}$.

Let $\varphi = \varphi(X_1,\dots,X_l)$ and $\psi = \psi(X_1,\dots,X_l)$ be
\MSO formulas. We say $\varphi$ and $\psi$ are \emph{equivalent}, written $\varphi
\equiv \psi$, if for all graphs $G$ and $U_1, \dots, U_l \subseteq V(G)$, $G
\models \varphi(U_1,\dots, U_l)$ if and only if $G \models \psi(U_1,\dots,
U_l)$.  Given a set $F$ of formulas, let ${F/\mathord\equiv}$ denote the set
of equivalence classes of $F$ with respect to $\equiv$. A system of
representatives of $F/\mathord\equiv$ is a set $R \subseteq F$ such that $R
\cap C \neq \emptyset$ for each equivalence class $C \in F/\mathord\equiv$.
The following statement has a straightforward proof using normal forms (see
\cite[Proposition~7.5]{Libkin04} for details).
\begin{fact}[\cite{GanianSlivovskySzeider13}]
\label{fact:representatives}
  Let $q$ and $l$ be fixed non\hy negative integers. The set
  $\MSO_{q,l}/\mathord\equiv$ is finite, and one can compute a system of
  representatives of $\MSO_{q,l}/\mathord\equiv$.
\end{fact}
We will assume that for any pair of non\hy negative integers $q$ and $l$ the
system of representatives of $\MSO_{q,l}/\mathord\equiv$ given by
Fact~\ref{fact:representatives} is fixed.

\newcommand{\introtypes}[0]{
\begin{definition}[\MSO Type]
  Let $q,l$ be non\hy negative integers. For a graph $G$ and an $l$\hy
  tuple $\vec{U}$ of sets of vertices of $G$, we define
  $\mathit{type}_q(G,\vec{U})$ as the set of formulas $\varphi \in
  \MSO_{q,l}$ such that $G \models \varphi(\vec{U})$. We call
  $\mathit{type}_q(G,\vec{U})$ the \MSO \emph{$q$-type of
    $\vec{U}$ in $G$}. 
\end{definition}
It follows from Fact~\ref{fact:representatives} that up to logical
equivalence, every type contains only finitely many formulas. The following Lemma~\ref{lem:typeformula} is obtained as an adaptation of a technical lemma from~\cite{GanianSlivovskySzeider13} to our setting, and allows us to represent types using \MSO formulas.

\begin{lemma}[see also~\cite{GanianSlivovskySzeider13}]
\label{lem:typeformula}
  Let $q$, $c$ and $l$ be non\hy negative integer constants, let $G$ be an $n$-vertex graph of rank-width at most $c$,
  and let $\vec{U}$ be an $l$\hy tuple of sets of vertices of $G$. One can
  compute a formula $\Phi \in \MSO_{q,l}$ such that for any graph
  $G'$ and any $l$\hy tuple $\vec{U}'$ of sets of vertices of $G'$ we have $G'
  \models \Phi(\vec{U}')$ if and only if $\mathit{type}_q(G,\vec{U}) =
  \mathit{type}_q(G',\vec{U}')$. Moreover, $\Phi$ can be computed in time $\bigoh(n^3)$.
\end{lemma}

\begin{proof}
  Let $R$ be a system of representatives of $\MSO_{q,l}/\mathord\equiv$ given
  by Fact~\ref{fact:representatives}. Because $q$ and $l$ are constants, we can
  consider both the cardinality of $R$ and the time required to compute it as
  constants. Let $\Phi \in \MSO_{q,l}$ be the formula defined as $\Phi =
  \bigwedge_{\varphi \in S} \varphi \wedge \bigwedge_{\varphi \in R \setminus
    S} \neg \varphi$, where $S = \SB \varphi \in R \SM G \models
  \varphi(\vec{U}) \SE$. We can compute $\Phi$ by deciding $G \models
  \varphi(\vec{U})$ for each $\varphi \in R$. Since the number of formulas in
  $R$ is a constant, this can be done in time $\bigoh(n^3)$ by applying Fact~\ref{fact:msorankwidth}.

  Let $G'$ be an arbitrary graph and let $\vec{U}'$ be an $l$\hy tuple of subsets of
  $V(G')$. We claim that $\mathit{type}_q(G, \vec{U}) = \mathit{type}_q(G',
  \vec{U'})$ if and only if $G' \models \Phi(\vec{U}')$. Since $\Phi \in
  \MSO_{q,l}$ the forward direction is trivial. For the converse, assume
  $\mathit{type}_q(G, \vec{U}) \neq \mathit{type}_q(G', \vec{U'})$. First
  suppose $\varphi \in \mathit{type}_q(G, \vec{U}) \setminus
  \mathit{type}_q(G', \vec{U'})$. The set $R$ is a system of representatives
  of $\MSO_{q,l}/\mathord\equiv$ , so there has to be a $\psi \in R$ such that
  $\psi \equiv \varphi$. But $G' \models \Phi(\vec{U}')$ implies $G' \models
  \psi(\vec{U}')$ by construction of $\Phi$ and thus $G' \models
  \varphi(\vec{U}')$, a contradiction. Now suppose $\varphi \in
  \mathit{type}_q(G', \vec{U}') \setminus \mathit{type}_q(G, \vec{U})$. An
  analogous argument proves that there has to be a $\psi \in R$ such that
  $\psi \equiv \varphi$ and $G' \models \neg \psi(\vec{U}')$. It follows that
  $G' \not \models \varphi(\vec{U}')$, which again yields a contradiction.
\end{proof}}
\lv{\introtypes} 

\newcommand{\MSOgames}[0]{
\begin{definition}[Partial isomorphism]\label{def:partialisomorphism}
  Let $G, G'$ be graphs, and let $\vec{V} = (V_1, \dots, V_l)$ and
  $\vec{U} = (U_1, \dots, U_l)$ be tuples of sets of vertices with
  $V_i \subseteq V(G)$ and $U_i \subseteq V(G')$ for each $i \in
  [l]$. Let $\vec{v} = (v_1, \dots, v_m)$ and $\vec{u} = (u_1, \dots,
  u_m)$ be tuples of vertices with $v_i \in V(G)$ and $u_i \in V(G')$
  for each $i \in [m]$. Then $(\vec{v}, \vec{u})$ defines a
  \emph{partial isomorphism between $(G, \vec{V})$ and $(G',
  \vec{U})$} if the following conditions hold:
  \begin{itemize}
    \item For every $i,j \in [m]$,
    \begin{align*}
      v_i = v_j \: \Leftrightarrow \: u_i = u_j \text{ and }
      v_iv_j \in E(G)\: \Leftrightarrow \: u_iu_j \in E(G').
    \end{align*}
    \item For every $i \in [m]$ and $j \in [l]$,
      \begin{align*}
        v_i \in V_j \: \Leftrightarrow u_i \in U_j.
      \end{align*}
    \end{itemize}
\end{definition}

\begin{definition}
  Let $G$ and $G'$ be graphs, and let $\vec{V_0}$ be a $k$\hy tuple of subsets
  of $V(G)$ and let $\vec{U_0}$ be a $k$\hy tuple of subsets of $V(G')$. Let
  $q$ be a non\hy negative integer. The \emph{$q$\hy round \MSO game on $G$
    and $G'$ starting from $(\vec{V_0}, \vec{U_0})$} is played as follows.
  The game proceeds in rounds, and each round consists of one of the following
  kinds of moves.
\begin{itemize}
  \item \textbf{Point move} The Spoiler picks a vertex in either $G$ or $G'$; the Duplicator responds by picking a vertex in the other graph.
  \item \textbf{Set move} The Spoiler picks a subset of $V(G)$ or a
    subset of $V(G')$; the Duplicator responds by picking a subset of the
    vertex set of the other graph.
  \end{itemize}
  Let $\vec{v}=(v_1,\dots,v_m), v_i \in V(G)$ and $
\vec{u}=(u_1,\dots,u_m), u_i \in V(G')$ be the point
  moves played in the $q$-round game, and let $\vec{V}=(V_1, \dots, V_l), V_i\subseteq V(G)$ and $\vec{U}=(U_1,\dots, U_l), U_i \subseteq V(G')$ be the set moves played in the
  $q$\hy round game, so that $l + m = q$ and moves belonging to same round
  have the same index. Then the Duplicator wins the game if $(\vec{v},
  \vec{u})$ is a partial isomorphism of $(G, \vec{V_0}\cup\vec{V})$ and $(G',
  \vec{U_0}\cup \vec{U})$. If the Duplicator has a winning strategy, we write $(G,
  \vec{V_0}) \equiv^{\MSO}_q (G', \vec{U_0})$.
\end{definition}

\begin{fact}[\cite{Libkin04}, Theorem 7.7]\label{thm:msogames} Given two graphs $G$ and $G'$ and two $l$\hy tuples $\vec{V_0}, \vec{U_0}$ of sets of vertices of $G$ and $G'$, we have \begin{center}
$    \mathit{type}_q(G, \vec{V_0}) = \mathit{type}_q(G, \vec{U_0}) \:
    \Leftrightarrow \: (G, \vec{V_0}) \equiv^{\MSO}_q (G', \vec{U_0}).$
\end{center}
\end{fact}}

\section{$(k,c)$-Well-Structured Modulators}\label{sec:kcwsm}

\begin{definition}
\label{def:wsm}
Let $\HH$ be a graph class and let $G$ be a graph. A set $\vec{X}$ of pairwise-disjoint split-modules of $G$ is called a $(k,c)$-\emph{\wsm}{} to $\HH$ if
\begin{enumerate}
\item $|\vec{X}|\leq k$, and
\item $\bigcup_{X_i\in \vec{X}} X_i$ is a modulator to $\HH$, and
\item $\rw(G[X_i])\leq c$ for each $X_i\in \vec{X}$.
\end{enumerate}
\end{definition}

For the sake of brevity and when clear from context, we will sometimes identify $\vec{X}$ with $\bigcup_{X_i\in\vec{X}} X_i$ (for instance $G-\vec{X}$ is shorthand for $G-\bigcup_{X_i\in\vec{X}} X_i$). 
To allow a concise description of our parameters, for any hereditary graph class $\HH$ we let the \emph{well-structure number} ($\wsn_c^{\HH}$ in short) denote the minimum $k$ such that $G$ has a $(k,c)$-{\wsm} to $\HH$. 

\begin{figure}[ht]
\centering
\begin{tikzpicture}[every node/.style={circle, fill=black, draw, scale=.3}, scale=0.5, rotate = 180, xscale = -0.8, yscale = 0.5]

\filldraw[fill opacity=0.7,fill=gray!20] (0,0) ellipse (3.7cm and 2.8cm);
\filldraw[fill opacity=0.7,fill=gray!20] (1.6,5) ellipse (2.2cm and 1.7cm);

 \foreach \x in {1,...,8}{%
   \pgfmathparse{(\x-1)*45+floor(\x/9)*45}
    \node (N-\x) at (\pgfmathresult:2.4cm) [thick] {};
 } 
 \foreach \x [count=\xi from 1] in {2,...,8}{%
    \foreach \y in {\x,...,8}{%
        \path (N-\xi) edge[-] (N-\y);
  }
 }

 \newlength{\gridsize}
\setlength{\gridsize}{0.3cm}

\node[above right=0.2cm and 1cm of N-1] (1)  {};
\node[right=\gridsize of 1] (2) {};
\node[below=\gridsize of 1] (3) {};
\node[below right=\gridsize and \gridsize of 2] (13) {};
\node[right=\gridsize and \gridsize of 2] (14) {};

\node[below=0.3cm of 3] (4) {};
\node[below=\gridsize of 4] (5) {};
\node[below right=\gridsize and \gridsize of 4] (6) {};
\node[above= \gridsize of 6] (15) {};
\node[above right=\gridsize and \gridsize of 6] (16) {};
\node[right=\gridsize of 6] (17) {};

\node[above right=0.5cm and 1cm of N-1] (7)  {};
\node[right=\gridsize of 7] (8) {};
\node[above=\gridsize of 7] (9) {};
\node[right=\gridsize of 8] (18) {};

\node[below=0.3cm of 5] (10) {};
\node[below=\gridsize of 10] (11) {};
\node[below right=\gridsize and \gridsize of 10] (12) {};
\node[above= \gridsize of 12] (19) {};
\node[above right=\gridsize and \gridsize of 12] (20) {};
\node[right=\gridsize of 12] (21) {};

\node[below left=0.45cm and 0.7cm of 5] (a1) {};
\node[below left=0.45cm and 0.7cm of 4] (a2) {};
\node[left=\gridsize of a1] (a3) {};
\node[left=\gridsize of a3] (a4) {};
\node[left=\gridsize of a2] (a5) {};
\node[left=\gridsize of a4] (a6) {};
\node[above=\gridsize of a4] (a7) {};

\draw (N-1)--(1)--(N-2);
\draw (N-2)--(3)--(N-1);
\draw (N-1)--(4)--(N-2);
\draw (N-1)--(7)--(N-2);

\draw (a3)--(a2)--(a1)--(a3)--(a4)--(a5)--(a3);
\draw (a4)--(a6)--(a7)--(a4);
\draw (a1)--(4)--(a2);
\draw (a1)--(5)--(a2);
\draw (a1)--(11)--(a2);

\draw (7)--(8)--(9);
\draw (2)--(1)--(3);
\draw (4)--(5)--(6);
\draw (10)--(11)--(12);
\draw (14)--(2)--(13);
\draw (15)--(6)--(16)--(6)--(17);
\draw (8)--(18);
\draw (13)--(16);
\draw (19)--(12)--(20)--(12)--(21);

\end{tikzpicture}
\caption{A graph with a $(2,1)$-\wsm{} to forests (in the two shaded areas).}
\label{fig:wsm}
\end{figure}

We conclude this section with a brief discussion on the choice of the parameter. The specific conditions restricting the contents of the modulator $\bigcup \vec{X}$ have been chosen as the most general means which allow both (1) the efficient finding of a suitable \wsm{}, and (2) the efficient use of this \wsm{} for kernelization. In this sense, we do not claim that there is anything inherently special about rank-width or split modules, other than being the most general notions which are currently known to allow the achievement of these two goals. 

In some of the applications of our results, we will consider graphs which have bounded expansion or bounded degree. We remark that in these cases, our results could equivalently be stated in terms of treewidth (instead of rank-width) and $\MSO_2$ logic (instead of $\MSO_1$ logic).\sv{~$(\star)$}
\lv{Details follow.}

\newcommand{\sparsity}[0]{
We say that a class $\HH$ of graphs is \emph{uniformly $k$-sparse} 
if there exists $k$ such that for every $G\in\HH$ every finite subgraph of $G$ has a number of edges bounded by $k$ times the number of vertices. 

\begin{fact}[\cite{Courcelle03}]\label{fact:MSOcolaps}
 For each integer $k$, one can effectively transform a given monadic
second-order formula using edge set quantifications into one that uses only vertex set
quantifications and is equivalent to the given one on finite, uniformly
$k$-sparse, simple, directed or undirected graphs. 
\end{fact}

\begin{fact}[\cite{Courcelle03}]\label{fact:twcwcolaps}
 A class of finite, uniformly $k$-sparse, simple, directed or undirected graphs
has bounded tree-width if and only if it has bounded clique-width.
\end{fact}

For definitions of \emph{Shallow minor}, \emph{Greatest reduced average density}, and \emph{bounded expansion} we refer to Definition 2.1, Definition 2.5, and Definition 2.6  in \cite{GajarskyHlinenyObdrzalek13}, respectively. 

\begin{observation}\label{obs:beissparse}
 Every class of graphs of bounded expansion is uniformly $k$-sparse for some positive integer constant $k$.
\end{observation}
\begin{proof}
 Let $\HH$ be a class of graphs of bounded expansion and let $f$ be the expansion function of $\HH$. Then $f(0)$ or equivalently the greatest reduced average density of $\HH$ with rank $0$ is constant and is an exact upper bound on the ratio between the number of edges and vertices of any subgraph of a graph in $\HH$. Therefore, $\HH$ is uniformly $f(0)$-sparse.
\end{proof}
}
\lv{\sparsity}

\section{A Case Study: Vertex Cover}
\label{sec:vc}

In this section we show how \wsm s to edgeless graphs can be used to obtain polynomial kernels for various problems. In particular, this special case can be viewed as a generalization of the vertex cover number. We begin by comparing the resulting parameter to known structural parameters.
Let $c\in \Nat$ be fixed and $\EE$ denote the class of edgeless graphs. The class $\ZZ$ containing only the empty graph will also be of importance later on in the section; we remark that while $\md^\ZZ$ represents a very weak parameter as it is equal to the order of the graph, this is not the case for $\wsn_c^\ZZ$. We begin by comparing \wsm s to edgeless graphs with similar parameters used in kernelization.

\lv{\begin{proposition}}
\sv{\begin{proposition}[$\star$]}
\label{prop:betterVC}
Let $\EE$ be graph class of edgeless graphs. Then:
\begin{enumerate}
\item $\rwc_c(G)\geq \wsn_c^\EE(G)$ for any graph $G$. Furthermore, for every $i\in \Nat$ there exists a graph $G_i$ such that $\rwc_c(G_i)\geq 2i$ and $\wsn_{c}^\EE=2$.
\item $\vcn(G)\geq \wsn_{1}^\EE(G)$ for any graph $G$. Furthermore, for every $i\in \Nat$ there exists a graph $G_i$ such that $\text{vcn}(G)\geq i$ and $\wsn_{1}^\EE=1$.
\end{enumerate}
\end{proposition}

\newcommand{\pfbetterVC}[0]{
\begin{proof}
The first claim follows from the fact that rank-width cover is also a \wsm{} to the empty graph. 
For the second claim, let $G'_c$ be a graph of rank-width $c+1$, of bounded degree and of order at least $i$ containing at least one vertex, say $v$, such that $G'_c-v$ has rank-width $c$. Next, we construct the graph $G_c$ from $G'_c-v$ by exhaustively applying the following operation: for each module in the graph containing more than a single vertex, we create a new pendant and attach it to a single vertex in that module. Observe that this operation preserves the rank-width of the graph, and moreover the resulting graph only contains trivial modules (i.e., modules which contain a single vertex). Finally, let $G^*_c$ be obtained from $2$ disjoint copies of $G_c$, say $G^1_c$ and $G^2_c$, and making the vertices which were adjacent to $v$ in $G^1_c$ adjacent to the vertices which were adjacent to $v$ in $G^2_c$. Then $\wsn_c^\EE(G^*_c)=2$, since $G^1_c$ and $G^2_c$ are each a \sm{} of rank-width at most $c$. However, since $G^*_c$ is a (vertex-)supergraph of $G'_c$, it follows that $\rw(G^*_c)\geq c+1$ and furthermore $G^*_c$ only contains trivial modules. Hence $\rwc_c(G^*_c)\geq 2i$.

The third claim follows from the fact that any vertex cover of $G$ is also a well-structured modulator to $\EE$. Finally, consider a path $P$ of length $2i+1$. Then $\vcn(P)\geq i$ but $\wsn^\EE_1(P)=1$.
\end{proof}}
\lv{\pfbetterVC}

It will be useful to observe that the above Proposition~\ref{prop:betterVC} also holds when restricted to the class of graphs of bounded expansion and bounded degree, and even when the graph class $\EE$ is replaced by $\ZZ$.

As we have established that already $\wsn^\EE_1\leq \vcn(G)$, it is important to mention that an additional structural restriction on the graph is necessary to allow the polynomial kernelization of $\MSOOPT{}{}$ problems in general (as is made explicit in the following Fact~\ref{fact:nokernel}).

\begin{fact}[\cite{BodlaenderJansenKratsch14}]
\label{fact:nokernel}
	\textsc{Clique} parameterized by the vertex cover number does not admit a polynomial kernel, unless \NP{} $\subseteq$ \text{\normalfont coNP/poly}. 
\end{fact}

However, it turns out that restricting the inputs to graphs of bounded expansion completely changes the situation: under this condition, it is not only the case that all all $\MSOMC{}$ and $\MSOOPT{}{}$ problems admit a linear kernel when parameterized by the vertex cover number, but also when parameterized by the more general parameter $\wsn_c^\EE$. To prove these claims, we begin by stating the following result.

\begin{fact}[\cite{GajarskyHlinenyObdrzalek13}]\label{fact:boundedexp}
Let $\KK$ be a graph class with bounded expansion. Suppose that for $G\in \KK$ and $S\in V(G)$, $\CC_1,\dots,\CC_s$ are sets of connected components of $G-S$ 
such that for all pairs $C, C'\in \cup_i\CC_i$ it holds that $C,C'\in \CC_j$ for some $j$ if and only if $N_S(C)=N_S(C')$. Let $\delta \ge 0$ be a constant bound on the diameter of these components, i.e., for all $C\in \cup_i\CC_i$, diam$(G[V(C)])\le \delta$.
Then there can be only at most $\bigoh(|S|)$ such sets $\CC_i$.
\end{fact}

This allows us to establish a key link between $\wsn_c^{\EE}$ and $\wsn_c^{\ZZ}$ on graphs of bounded expansion.

\begin{lemma}\label{lem:edgeless-empty}
 Let $\KK$ be a graph class with bounded expansion.
 Then there exists a constant $d$ such that for every $G\in \KK$ it holds that $\wsn_c^{\ZZ}(G) \le d\cdot(\wsn_c^{\EE}(G))$.
\end{lemma}

\begin{proof}
 Let $k=\wsn_c^{\EE}(G)$ and let $\vec{H}$ be a $(k,c)$-\wsm{} to $\EE$. Let $S$ be a set of vertices containing exactly one vertex from the frontier of every \sm{} in $\vec{H}$.
 The graph $G'=G-(\vec{H}-S)$ is a graph with bounded expansion and $S$ is its vertex cover. Clearly, the diameter of every connected component of $G'\setminus S$ is at most $1$ (every connected component is a singleton). Therefore, by Fact~\ref{fact:boundedexp} there exists a constant $d'$ such that there are at most $d'\cdot|S|=d'\cdot\wsn_c^{\EE}(G)$ sets of vertices $\CC_1,\dots,\CC_s$ in $G'-S$ 
such that for all pairs $v, v'\in \cup_i\CC_i$ it holds that $v,v'\in \CC_j$ for some $j$ if and only if $N_S(v)=N_S(v')$.
Clearly each such $\CC_i$ is a \sm{} in $G'$, and hence also in $G$. Furthermore, each such $\CC_i$ has rank-width at most $1$. Hence $\wsn_c^\ZZ(G)\leq \wsn_c^\EE(G)+d'\cdot \wsn_c^{\EE}(G)$.
\end{proof}

The above lemma allows us to shift our attention from modulators to $\EE$ to a partition of the vertex set into \sm s of bounded rank-width. The rest of this section is then dedicated to proving our results for \wsm s to $\ZZ$. Our proof strategy for this special case of \wsm s closely follows the replacement techniques used to obtain the kernelization results for the rank-width cover number~\cite{GanianSlivovskySzeider13}, with the distinction that many of the tools and techniques had to be generalized to cover splits instead of modules.

\newcommand{\MSOVC}[0]{
\begin{fact}
[\cite{EibenGanianSzeider15}]
\label{fact:constantrep}
Let $q$, $c$ be non\hy negative integer
  constants. Let $G$ be an $n$-graph of rank-width at most $c$ and $S\subseteq V(G)$. Then one can in time $\bigoh(n^3)$ compute a graph $G'$ and a set $S'\subseteq V(G')$ such that
  $\Card{V(G')}$ is bounded by a constant and $\mathit{type}_q(G,S) =
  \mathit{type}_q(G',S')$.
\end{fact}

We use the notion of \emph{similarity}~\cite{EibenGanianSzeider15} to prove that this procedure does not change the outcome of \MSOMC{\varphi}.

\begin{definition}\label{def:similarity}
  Let $q$ and $k$ be non\hy negative integers, $\HH$ be a graph class, and let $G$ and $G'$ be graphs with $(k,c)$-\wsm s $\vec{X}=\{X_1,\dots,X_k\}$ and $\vec{X'}=\{X_1',\dots,X_k'\}$ to $\HH$,
  respectively. For $1\leq i\leq k$, let $S_i=\lambda(X_i)$ and similarly let $S'_i=\lambda(X'_i)$.
We say that $(G, \vec{X})$ and $(G', \vec{X}')$ are \emph{$q$-similar}
  if all of the following conditions are met:
  \begin{enumerate} 
  \item There exists an isomorphism $\tau$ between $G-\vec{X}$ and $G'-\vec{X}'$. \label{cond:sameH}
  \item For every $v\in V(G)\setminus \vec{X}$ and $i\in [k]$, it holds that $v$ is adjacent to $S_i$ if and only if $\tau(v)$ is adjacent to $S'_i$.\label{cond:Hcon}
  \item if $k\geq 2$, then for every $1\leq i<j\leq k$ it holds that $S_i$ and $S_j$ are adjacent if and only if $S'_i$ and $S'_j$ are adjacent. \label{cond:splitcon}
  \item For each $i\in [k]$, it holds that $\mathit{type}_q(G[X_i],S_i)=\mathit{type}_q(G'[X'_i],S'_i)$. \label{cond:sametypes}
\end{enumerate}
\end{definition}

\begin{lemma}\label{lem:similar}
  Let $q$, $c$ be non\hy negative integer constants and $\HH$ be a graph class. Then given an $n$-vertex graph $G$ and a $(k,c)$-\wsm{} $\vec{X}=\{X_1,\dots X_k\}$ of $G$ into $\HH$, one can in time $\bigoh(n^3)$ compute a graph $G'$ with a $(k,c)$-\wsm{} $\vec{X'}=\{X'_1,\dots X'_k\}$ into $\HH$ such that $(G,\vec{X})$ and $(G',\vec{X'})$ are $q$-similar and for each $i \in [k]$ it holds that $|X_i'|$ is bounded by a constant.
\end{lemma}

\begin{proof} 
  For $i\in [k]$, let $S_i=\lambda(X_i)$, $G_i=G[X_i]$, and let $G_0=G\setminus G[\vec{X}]$. We compute a graph $G_i'$ of constant size and a set $S_i'\subseteq V(G_i')$ with the same \MSO $q$-type as $(G_i,S_i)$. By Fact~\ref{fact:constantrep}, all of this can be done in 
  time $\bigoh(n^3)$. Now let $G'$ be the graph obtained by the following procedure:
  
  \begin{enumerate}
  \item Perform a disjoint union of $G_0$ and $G'_i$ for each $i\in [k]$;
  \item If $k\geq 2$ then for each $1\leq i<j \leq k$ such that $S_i$ and $S_k$ are adjacent in $G$, we add edges between every $v\in S'_i$ and $w\in S'_j$.
  \item for every $v\in V(G_0)$ and $i\in [k]$ such that $S_i$ and $\{v\}$ are adjacent, we add edges between $v$ and every $w\in S'_i$.
  \end{enumerate}
  
   It is easy to verify that $(G,\vec{X})$ and $(G', \vec{X}')$, where $\vec{X}'=\{V(G_1'),\dots,V(G_k')\}$, are $q$\hy similar.
\end{proof}

\begin{fact}[\cite{EibenGanianSzeider15}]
\label{fact:partitiongame}
  Let $q$, $c$, and $k$ be non\hy negative integers, $\HH$ be a graph class, and let $G$ and $G'$ be graphs with $(k,c)$-\wsm s 
  $\vec{X}=\{X_1,\dots,X_k\}$ and $\vec{X'}=\{X_1',\dots,X_k'\}$ to $\HH$,
  respectively. If $(G, \vec{X})$ and $(G', \vec{X}')$ are \emph{$q$-similar}, then 
  $\mathit{type}_q(G, \emptyset) = \mathit{type}_q(G',
  \emptyset)$.
  \end{fact}}

\lv{\MSOVC}
  
\lv{\begin{theorem}}
\sv{\begin{theorem}[$\star$]}
\label{thm:msovc}
Let $\KK$ be a graph class of bounded expansion, $\EE$ be the class of edgeless graphs and $\ZZ$ be the class of empty graphs. 
For every $\MSO$ sentence $\varphi$ 
the problem $\MSOMC{\varphi}$ admits a linear kernel parameterized by $\wsn_c^{\ZZ}$. Furthermore, the problem $\MSOMC{\varphi}$ admits a linear kernel parameterized by $\wsn_c^\EE$ on $\KK$.
\end{theorem}

\newcommand{\pfmsovc}[0]{
\begin{proof}
 By Lemma~\ref{lem:edgeless-empty} it is sufficient to show that $\MSOMC{\phi}$ admits a linear kernel parameterized by $\wsn_c^{\ZZ}$.
 Let $G$ be a graph, $k=\wsn_c^{\ZZ}(G)$, and $q$ be the nesting depth of quantifiers in $\phi$. 
 By Fact~\ref{fact:equiv} we can find the set $\vec{X}$ of equivalence classes of $\sim_c^G$ in polynomial time. Clearly, the set $\vec{X}$ is a $(k,c)$-\wsm{} to the empty graph. We proceed by constructing $(G',\vec{X}')$ by Lemma~\ref{lem:similar}. Since each $X'_i\in\vec{X}'$ has size bounded by a constant, $|\vec{X}'|\leq k$, and $\bigcup\vec{X}' = V(G')$, it follows that $G'$ is an instance of $\MSOMC{\varphi}$ of size $\bigoh(k)$. Finally, since $G$ and $G'$ are $q$-similar, it follows from Fact~\ref{fact:partitiongame} that 
$G\models \phi$ if and only if $G'\models \phi$. 
\end{proof}}
\lv{\pfmsovc}
\sv{\begin{proof}[Sketch of Proof]
 By Lemma~\ref{lem:edgeless-empty} it is sufficient to show that $\MSOMC{\phi}$ admits a linear kernel parameterized by $\wsn_c^{\ZZ}$.
 Let $G$ be a graph, $k=\wsn_c^{\ZZ}(G)$, and $q$ be the nesting depth of quantifiers in $\phi$. 
 By Fact~\ref{fact:equiv} we can find the set $\vec{X}$ of equivalence classes of $\sim_c^G$ in polynomial time. Clearly, the set $\vec{X}$ is a $(k,c)$-\wsm{} to the empty graph. We proceed by using replacement techniques to construct an equivalent graph $(G',\vec{X}')$ such that each $X'_i\in\vec{X}'$ has size bounded by a constant. Since $|\vec{X}'|\leq k$ and $\bigcup\vec{X}' = V(G')$, it follows that $G'$ is an instance of $\MSOMC{\varphi}$ of size $\bigoh(k)$. 
\end{proof}}

\newcommand{\aMSOOPT}[2]{a\textsc{MSO-Opt${}^{#1}_{#2}$}} %

Next, we combine the approaches used in~\cite{GanianSlivovskySzeider13} and~\cite{EibenGanianSzeider15} to handle \MSOOPT{\gleq}{\varphi} problems by using our more general parameters. Similarly as in~\cite{GanianSlivovskySzeider13}, we use a more involved replacement procedure which explicitly keeps track of the original cardinalities of sets and results in an \emph{annotated version} of $\MSOOPT{\gleq}{\varphi}$. However, some parts of the framework (in particular the replacement procedure) had to be reworked using the techniques developed in~\cite{EibenGanianSzeider15}, since we now use \sm s instead of simple modules. 

\newcommand{\annotation}[0]{
Given a graph $G=(V,E)$, an \emph{annotation} $\WWW$ is a set of
triples $(X,Y,w)$ with $X \subseteq V, Y\subseteq V, w \in \Nat$. For every set $A \subseteq
V$ we define
\begin{center}
${\WWW}(A)=\sum_{(X,Y,w)
  \in {\WWW}, X\subseteq Z, Y\cap Z =\emptyset} w.$
\end{center}
The idea is that a triple $(X, Y, w)$ assigns weight $w$ to a vertex set
$X$. Specifying the set $Y$ allows us to control which subsets of $Z$ the
above sum is taken over. In the kernel, each set $X$ will be a subset of a
module $M$ (with weight $w$ corresponding to the optimum cardinality of a set
in the matching module of the original graph). Setting $Y = M \setminus X$
ensures that the sum $\WWW(Z)$ contains at most one term for each module $M$.
Note that an instance of $\MSOOPT{\gleq}{\varphi}$ can be represented as an
instance of $\aMSOOPT{\gleq}{\varphi}$ with the annotation
${\WWW}=\SB(\{v\},\emptyset,1) \SM v\in V(G)\SE$.  We call the pair $(G,
\WWW)$ an \emph{annotated graph}. If the integer $w$ is represented in binary,
we can represent a triple $(X,Y,w)$ in space $\Card{X}+\Card{Y}+\log_2(w)$.
Consequently, we may assume that the size of the encoding of an annotated
graph $(G,\WWW)$ is polynomial in $\Card{V(G)}+\Card{\WWW}+\max_{(X,Y,w)\in
  \WWW} \log_2 w$.  Each \MSO formula $\varphi(X)$ and $\gleq\in\{\leq,
\geq\}$ gives rise to an \emph{annotated MSO-optimization problem}.
\begin{quote}
  $\aMSOOPT{\gleq}{\varphi}$\\
  \nopagebreak \emph{Instance}: A graph $G$ with an annotation $\WWW$ and an integer~$r\in
  \Nat$. \\ \nopagebreak \nopagebreak \emph{Question}: Is there a set
  $Z\subseteq V(G)$ such that $G \models \varphi(Z)$ and ${\WWW}(Z)\gleq r$?
\end{quote}

\begin{lemma}\label{lem:anngraph}
  Let $\varphi = \varphi(X)$ be a fixed $\MSO$ formula.
  Then given an instance $(G, r)$ of $\MSOOPT{\le}{\varphi}$ and a $(k,c)$-\wsm $\vec{X} = {X_1, \dots , X_k }$ to $\ZZ$ of $G$, an annotated
graph $(G' , \WW)$ satisfying the following properties can be computed in polynomial time.
\begin{enumerate}
 \item $(G,r)\in \MSOOPT{\le}{\varphi}$ if and only if $(G', \WW,r)\in a\MSOOPT{\le}{\varphi}$.
 \item $|V(G')|\in \bigoh(k)$.
 \item The encoding size of $(G' , \WW)$ is $\bigoh(k \log(|V(G)|))$.
\end{enumerate}
\end{lemma}
\begin{proof}
  The proof proceeds analogously to the proof of Theorem~********* in~\cite{GanianSlivovskySzeider13}, with the distinction that we use the replacement procedure described in Fact~\ref{fact:constantrep}. 
  
In particular, using Lemma~\ref{lem:similar} we compute a graph $G'$ with a $(k,c)$-\wsm{} $\{X_1',\dots,X_k'\}$ to $\ZZ$ such that $(G, \vec{X})$ and $(G', \vec{X}')$ are $(q+1)$\hy similar and $|X_i'|$ is bounded by a constant for each $i\in [k]$. To compute the   annotation $\WWW$, we proceed as follows.  For each $i \in [k]$, we go    through all subsets $W' \subseteq X_i'$. By Lemma~\ref{lem:typeformula}, we  can compute a formula $\Phi$ such that for any graph $H$ and $W \subseteq   V(H)$ we have $\mathit{type}_q(G'[X_i'], W') = \mathit{type}_q(H, W)$ if and   only if $H \models \Phi(W)$. Since $\Card{X_i'}$ has constant size for every $i \in [k]$, this can be done within a constant time bound. Moreover, since $(G, \vec{X})$ and $(G', \vec{X}')$ are   $(q+1)$\hy similar, there has to exist a $W \subseteq X_i$ such that $G[X_i]   \models \Phi(W)$.  Using Fact~\ref{fact:msorankwidth}, we can compute a minimum\hy cardinality subset $W^* \subseteq X_i$ with this property in polynomial time. We then add the triple $(W', X_i' \setminus W',  \Card{W^*})$ to $\WWW$. In total, the number of subsets processed is in  $O(k)$. From this observation we get the desired bounds on the total  runtime, $\Card{V(G')}$, and the encoding size of $(G', \WWW)$.

  We claim that $(G', \WWW, r) \in \aMSOOPT{\leq}{\varphi}$ if and only if  $(G, r) \in \MSOOPT{\leq}{\varphi}$. Suppose there is a set $W \subseteq  V(G)$ of vertices such that $G \models \varphi(W)$ and $\Card{W} \leq  r$. Since $X_1,\dots,X_k$ is a partition of $V(G)$, we have $W = \bigcup_{i \in    [k]} W_i$, where $W_i = W \cap X_i$. For each $i \in [k]$, let $W_i^*  \subseteq X_i$ be a subset of minimum cardinality such that  $\mathit{type}_q(G[X_i], W_i) = \mathit{type}_q(G[X_i], W_i^*)$. From the $(q+1)$\hy similarity of $(G, \vec{U})$ and  $(G', \vec{U}')$, there is $W_i' \subseteq X_i'$ for each $i \in [k]$ such  that $\mathit{type}_q(G'[X_i'], W_i') = \mathit{type}_q(G[X_i], W_i^*)$. By  construction, $\WWW$ contains a triple $(W_i', X_i' \setminus W_i',  \Card{W_i^*})$. Observe that $(X, Y, w) \in \WWW$ and $(X, Y, w') \in \WWW$  implies $w = w'$. Let $W' = \cup_{i \in [k]} W_i'$. Then by $(q+1)$\hy  similarity of $(G, \vec{X})$ and $(G', \vec{X}')$ and  Fact~\ref{fact:partitiongame}, we must have $\mathit{type}_q(G, W) =  \mathit{type}_q(G', W')$. In particular, $G' \models   \varphi(W')$. Furthermore,
\begin{center}
$
     \WWW(W') = \sum_{(W_i', X_i' \setminus W_i', \Card{W_i^*}) \in \WWW, X_i'
       \cap W' = W'_i} \Card{W_i^*} \leq \sum_{i \in [k]} \Card{W_i} = \Card{W}
     \leq r.
$ \end{center}
   
   For the converse, let $W' \subseteq V(G')$ such that $\WWW(W') \leq r$ and
   $G' \models \varphi(W')$, let $W_i'$ denote $W' \cap X_i'$ for $i \in [k]$.
   By construction, there is a set $W_i \subseteq
   X_i$ for each $i \in [k]$ such that $\mathit{type}_q(G[X_i], W_i) =
   \mathit{type}_q(G'[X_i'], W'_i)$ and $\WWW(W') = \sum_{i \in [k]}
   \Card{W_i}$. Let $W = \cup_{i \in [k]} W_i$. Then by congruence and
   Fact~\ref{fact:partitiongame} we get $\mathit{type}_q(G, W) =
   \mathit{type}_q(G', W')$ and thus $G \models \varphi(W)$. Moreover,
   $\Card{W} = \WWW(W') \leq r$.

\end{proof}

To complete the proof, we will make use of a win-win argument based on the following fact.

\begin{fact}[Folklore]\label{fact:decidingMSO}
  Given an $\MSO$ sentence $\varphi$ and a graph $G$, one can decide whether $G \models \varphi$ in time $\bigoh(2^{nl})$, where $n = |V (G)|$ and $l = |\varphi|$.
\end{fact}
}
\lv{\annotation}

\lv{\begin{theorem}}
\sv{\begin{theorem}[$\star$]}
\label{thm:ann}
Let $\EE$ be a class of edgeless graphs and $\ZZ$ be the class containing the empty graph. For every $\MSO$ formula $\varphi$ 
the problem $\MSOOPT{\leq}{\varphi}$ admits a linear bikernel parameterized by $\wsn_c^{\EE}$ on any class of graphs of bounded expansion, and a linear bikernel parameterized by $\wsn_c^{\ZZ}$.
\end{theorem}

\newcommand{\pfann}[0]{
\begin{proof}
 By Lemma~\ref{lem:edgeless-empty} it is sufficient to show that $\MSOMC{\phi}$ admits a linear bikernel parameterized by $\wsn_c^{\ZZ}$.
 By Fact~\ref{fact:equiv} we can find equivalence classes $\vec{X}=\{X_1,\dots,X_k\}$ of $\sim_c$ in polynomial time. 
  
 Let $(G', \WW)$ be the annotated graph computed from $G$ and $\vec{X}$ according to Lemma~\ref{lem:anngraph}.
 Let $n=|V(G)|$ and suppose $2^k\le n$. Then we can solve $(G',\WW,r)$ in time $n^{\bigoh(1)}$. To do this, we go through all $2^{\bigoh(k)}$ subsets $W$ of $G'$ and test whether $\WW(W)\le r$. If that is the case, we check whether $G' \models \varphi(W )$. By Fact~\ref{fact:decidingMSO} this check can be carried out in time $c_12^{c_2k}\le c_1 n^{c2}$
for suitable constants $c_1$ and $c_2$ depending only on $\varphi$. Thus we can find a constant $t$ such that the entire procedure
runs in time $n^t$ whenever $n$ is large enough.  If we find a solution $W\subseteq V(G')$  we return a trivial yes-instance;
otherwise, a trivial no-instance (of $a\MSOOPT{\leq}{\varphi}$).  Now suppose $n < 2^k$. Then $log(n) < k$ and so the encoding
size of $\WW$ is polynomial in $k$. Thus $(G',\WW,r)$ is a polynomial bikernel.
\end{proof}
}
\lv{\pfann}

\section{Finding $(k,c)$-Well-Structured Modulators}
\label{sec:finding}

For the following considerations, we fix $c$ and assume that the graph $G$ has rank-width at least $c+2$ (this is important for Fact~\ref{fact:equiv}). This assumption is sound, since the considered problems can be solved in polynomial time on graphs of bounded rank-width. Recall that given a \sm{} $A$ in $G$, we use $\lambda(A)$ to denote the frontier of $A$. This section will show how to efficiently approximate \wsm s to various graph classes; in particular, we give algorithms for the class of forests and then for any graph class which can be characterized by a finite set of forbidden induced subgraphs.

The following Fact~\ref{fact:equiv} linking rank-width and \sm s will be crucial for approximating our \wsm s.

\begin{definition}
  Let $G$ be a graph and $c \in \Nat$.  We define a relation
  $\sim^G_c$ on $V(G)$ by letting $v \sim^G_c w$ if and only if there
  is a {\sm} $M$ of $G$ with $v,w \in M$ and $\rw(G[M]) \leq c$. We
  drop the superscript from $\sim^G_c$ if the graph $G$ is clear from
  context.
\end{definition}

\begin{fact}[\cite{EibenGanianSzeider15}]
\label{fact:equiv}
Let $c\in \Nat$ be fixed and $G$ be a graph of rank-width at least $c+2$. The relation $\sim^G_c$ is an equivalence, and any graph $G$ has its vertex set uniquely partitioned by the equivalence classes of $\sim_c$ into inclusion-maximal \sm s of rank-width at most $c$. Furthermore, for $a,b\in V(G)$ it is possible to test $a\sim_c b$ in $\bigoh(n^3)$ time.
\end{fact}

\subsection{Finding $(k,c)$-Well-Structured Modulators to Forests}
\label{sub:fvs}

Our starting point is the following lemma, which shows that long cycles which hit a non-singleton frontier imply the existence of short cycles. 

\lv{\begin{lemma}}
\sv{\begin{lemma}[$\star$]}
\label{lem:longcycle}
Let $C$ be a cycle in $G$ such that $C$ intersects at least three distinct equivalence classes of $\sim_c$, one of which has a frontier of cardinality at least $2$. Let $Z$ be the set of equivalence classes of $\sim_c$ which intersect $C$. Then there exists a cycle $C'$ such that the set $Z'$ of equivalence classes it intersects is a subset of $Z$ and has cardinality at most $3$.
\end{lemma}	

\newcommand{\pflongcycle}[0]{
\begin{proof}
Let $B$ be an equivalence class in $Z$ such that $b_1,b_2\in \lambda(B)$ are two distinct vertices. By assumption, $C$ must contain two distinct vertices $a,c\not \in B$ which are adjacent to $\lambda(B)$. Then $a,b_1,c,b_2$ forms a $C_4$ in $G[B\cup \{a,c\}]$.
\end{proof}}
\lv{\pflongcycle}

We will use the following observation to proceed when Lemma~\ref{lem:longcycle} cannot be applied.

\begin{observation}
\label{obs:fvs}
Assume that for each equivalence class $B$ of $\sim_c$ it holds that $G[B]$ is acyclic, and that no cycle intersects $B$ if $|\lambda(B)|\geq 2$. Then for every cycle $C$ in $G$ and every vertex $a\in C$, it holds that $a$ is in the frontier of some equivalence class of $\sim_c$.
\end{observation}

Fact~\ref{fact:fvs} below is the last ingredient needed for the algorithm.

\begin{fact}[\cite{BeckerGeiger96}]
\label{fact:fvs}
\textsc{Feedback Vertex Set} can be $2$-approximated in polynomial time.
\end{fact}

\lv{\begin{theorem}}
\sv{\begin{theorem}[$\star$]}
\label{thm:findwsmforest}
Let $c\in \Nat$ and $\FF$ be the class of forests. There exists a polynomial algorithm which takes as input a graph $G$ of rank-width at least $c+2$ and computes a set $\vec{X}$ of \sm s such that $\vec{X}$ is a $(k,c)$-well-structured modulator to $\FF$ and $k\leq 3\cdot \wsn_c^\FF$.
\end{theorem}

\newcommand{\pffindwsmforest}[0]{
\begin{proof}
We first describe the algorithm and then argue correctness. 
The algorithm proceeds in three steps. 
\begin{enumerate}
\item[I] By deciding $a\sim_c b$ for each pair of vertices in $G$ as per Fact~\ref{fact:equiv}, we compute the equivalence classes of $\sim_c$. 
\item[II] For each set of up to three equivalence classes $\{A_1, A_2,A_3\}$ of $\sim_c$, we check if $G[A_1\cup A_2\cup A_3]$ is acyclic; if it's not, then we add $A_1$, $A_2$ and $A_3$ to $\vec{X}$ and set $G:=G-(A_1\cup A_2 \cup A_3)$. 
\item[III] We use Fact~\ref{fact:fvs} to $2$-approximate a feedback vertex set $S$ of $G$ in polynomial time; let $S'$ contain every equivalence class of $\sim_c$ which intersects $S$. We then set $\vec{X}:=\vec{X}\cup S'$, and output $\vec{X}$.
\end{enumerate}

For correctness, observe that Step III guarantees that $G-\vec{X}$ is acyclic. Hence we only need to argue that $|\vec{X}|\leq 3\cdot \wsn_c^\FF$. So, assume for a contradiction that there exists a $(k,c)$-\wsm{} $\vec{X'}$ to $\FF$ such that $|\vec{X}|>3\cdot k$. Let $\Lambda$ be the set of all equivalence classes of $\sim_c$ which were added to $\vec{X}$ in Step II of the algorithm. Since for each such $\{A_1, A_2,A_3\}$ the graph $G[A_1\cup A_2\cup A_3]$ contains a cycle and $A_1$, $A_2$, $A_3$ are inclusion-maximal \sm s of rank-width at most $c$ by Fact~\ref{fact:equiv}, $\vec{X}'$ must always contain at least one \sm $A'$ such that $A'\subseteq A_i$ for some $i\in [3]$. Let $\Lambda'$ contain all such \sm s $A'$, i.e., all elements of $\vec{X}'$ which form a subset of a \sm{} added to $\vec{X}$ in Step II.

Let $\vec{X}_3=\vec{X}\setminus \Lambda$ and $\vec{X}'_3=\vec{X}'\setminus \Lambda'$. Since $|\Lambda|\leq 3\cdot |\Lambda'|$ by the argument above, from our assumption it would follow that $|\vec{X}_3|>3\cdot |\vec{X}'_3|$. Let us consider the graphs $G_3=G-\Lambda$ and $G_3'=G'-\Lambda'$; observe that $G_3\subseteq G_3'$. Furthermore, by Lemma~\ref{lem:longcycle} a cycle $C$ in $G_3$ cannot intersect any equivalence class $B$ of $\sim_c$ such that $\lambda(B)\geq 2$. Hence we can apply Observartion~\ref{obs:fvs}, from which it follows that there is a one-to-one correspondence between any minimal feedback vertex set in $G_3$ and the equivalence classes of $\sim_c$ in $G_3$. Let $z$ be the cardinality of a minimum feedback vertex set in $G_3$; by the correctness of the algorithm of Fact~\ref{fact:fvs}, we have $z\leq |\vec{X}_3|\leq 2z$. Since $G_3'$ is a supergraph of $G_3$, it follows that $|\vec{X}'_3|\geq z$, and hence from our assumption we would obtain $2z\geq |\vec{X}_3|>3\cdot |\vec{X}'_3|\geq 3z$. We have thus reached a contradiction, and conclude that there exists no $(k,c)$-\wsm{} to $\FF$ such that $|\vec{X}|>3\cdot k$.
\end{proof}
}
\lv{\pffindwsmforest}

\sv{
\begin{proof}[Sketch of Proof]
The algorithm proceeds in three steps. 
\begin{enumerate}
\item[I] By deciding $a\sim_c b$ for each pair of vertices in $G$ as per Fact~\ref{fact:equiv}, we compute the equivalence classes of $\sim_c$. 
\item[II] For each set of up to three equivalence classes $\{A_1, A_2,A_3\}$ of $\sim_c$, we check if $G[A_1\cup A_2\cup A_3]$ is acyclic; if it's not, then we add $A_1$, $A_2$ and $A_3$ to $\vec{X}$ and set $G:=G-(A_1\cup A_2 \cup A_3)$. 
\item[III] We use Fact~\ref{fact:fvs} to $2$-approximate a feedback vertex set $S$ of $G$ in polynomial time; let $S'$ contain every equivalence class of $\sim_c$ which intersects $S$. We then set $\vec{X}:=\vec{X}\cup S'$, and output $\vec{X}$.\qedhere
\end{enumerate}
\end{proof}
}

\subsection{Finding $(k,c)$-Well-Structured Modulators via Obstructions}
\label{sub:fis}

Here we will show how to efficiently find a sufficiently small $(k,c)$-\wsm{} to any graph class which can be characterized by a finite set of forbidden induced subgraphs. Let us fix a graph class $\HH$ characterized by a set $\RR$ of forbidden induced subgraphs, and let $r$ be the maximum order of a graph in $\RR$. Our first step is to reduce our problem to the classical \textsc{Hitting Set} problem, the definition of which is recalled below.

\begin{center}
\vspace{-0.5cm}
  \begin{boxedminipage}[t]{0.99\textwidth}
\begin{quote}
\smallskip
  $d$-\textsc{Hitting Set}\\ \nopagebreak
  \emph{Instance}: A ground set $S$ and a collection $\CC$ of subsets of $S$, each of cardinality at most $d$.\\ \nopagebreak
  \emph{Notation}: A hitting set is a subset of $S$ which intersects each set in $\CC$.\\ \nopagebreak
  \emph{Task}: Find a minimum-cardinality hitting set.
   \smallskip
\end{quote}
\end{boxedminipage}
\end{center}

Given a graph $G$ (of rank-width at least $c+2$), we construct an instance $W_G$ of $r$-\textsc{Hitting Set} as follows. The ground set of $W$ contains each equivalence class $A\subseteq V(G)$ of $\sim_c$. For each induced subgraph $R\subseteq G$ isomorphic to an element of $\RR$, we add the set $C_R$ of equivalence classes of $\sim_c$ which intersect $R$ into $\CC$. This completes the construction of $W_G$; we let $\hit(W_G)$ denote the cardinality of a solution of $W_G$.

\lv{\begin{lemma}}
\sv{\begin{lemma}[$\star$]}
\label{lem:hitting}
For any graph $G$ of rank-width at least $c+2$, the instance $W_G$ is unique and can be constructed in polynomial time. Every hitting set $Y$ in $W_G$ is a $(|Y|,c)$-\wsm{} to $\HH$ in $G$.
Moreover, $\wsn_c^\HH=\hit(W_G)$.
\end{lemma}

\newcommand{\pflemhitting}[0]{
\begin{proof}
The uniqueness of $W_G$, as well as the fact that $W_G$ can be constructed in polynomial time, follow from Fact~\ref{fact:equiv} together with the observation that all subgraphs $R\subseteq G$ isomorphic to an element of $\RR$ can be enumerated in polynomial time. 

For the second claim, consider a hitting set $Y\subseteq S$. The graph $G-Y$ cannot contain any obstruction for $\HH$, and hence $G-Y\in \HH$. 

For the third claim, assume $G$ contains a $(\hit(W_G)-1,c)$-\wsm{} $\vec{X}$ to $\HH$. By Fact~\ref{fact:equiv}, each element $A$ of $\vec{X}$ forms a subset of an equivalence class $A'$ of $\sim_c$. Let $\vec{X}'$ be obtained by replacing each element of $A$ by its respective superset $A'$. Then $G-\vec{X}'$ is a subgraph of $G-\vec{X}$, and hence by our assumption $G-\vec{X}'$ cannot contain any subgraph isomorphic to an element of $\RR$. However, this would imply that $\vec{X}'$ is a hitting set in $W_G$ of cardinality $\hit(W_G)-1$, which is a contradiction.
\end{proof}
}
\lv{\pflemhitting}

The final ingredient we need for our approximation algorithm is the following folklore result.

\begin{fact}[Folklore, see also~\cite{KleinbergTardos06}]
\label{fact:hit}
There exists a polynomial-time algorithm which takes as input an instance $W$ of $r$-\textsc{Hitting Set} and outputs a hitting set $Y$ of cardinality at most $r\cdot \hit(W)$.
\end{fact}

\begin{theorem}
\label{thm:findwsm}
Let $c\in \Nat$ and $\HH$ be a class of graphs characterized by a finite set of forbidden induced subgraphs of order at most $r$. There exists a polynomial algorithm which takes as input a graph $G$ of rank-width at least $c+2$ and computes a $(k,c)$-well-structured modulator to $\HH$ such that $k\leq r\cdot \wsn_c^\HH$.
\end{theorem}

\begin{proof}
We proceed in two steps: first, we compute the $r$-\textsc{Hitting Set} instance $W_G$, and then we use Fact~\ref{fact:hit} to compute an $r$-approximate solution $Y$ of $W_G$ in polynomial time. We then set $\vec{X}:=Y$ and output.
Correctness follows from Lemma~\ref{lem:hitting}. \lv{In particular, the hitting set $Y$ computed by Fact~\ref{fact:hit} has cardinality at most $r\cdot \hit(W_G)$ and hence $|\vec{X}|\leq r\cdot \wsn_c^\HH(G)$ by Lemma~\ref{lem:hitting}.}
\end{proof}

\section{Applications of $(k,c)$-Well-Structured Modulators}
\label{sec:using}

We now proceed by outlining the general applications of our results. Our algorithmic framework is captured by the following Theorem~\ref{thm:main-use}.

\begin{theorem}\label{thm:main-use}
Let $p, q$ be polynomial functions. For every $\MSO$ sentence $\phi$ and every graph class $\HH$ such that 
\begin{enumerate}
\item \label{cond1}$\MSOMC{\phi}$ admits a (bi)kernel of size $p(\md^\HH(G))$, and
\item \label{cond2}there exists a polynomial algorithm which finds a $(q(\wsn^\HH_c),c)$-\wsm{} to $\HH$,
\end{enumerate} 
the problem $\MSOMC{\phi}$ admits a (bi)kernel of size $p(q(\wsn_c^\HH(G)))$.
\end{theorem}

\begin{proof}
Let $G$ be a graph, $k=\wsn_c^\HH(G)$ and $s$ be the nesting depth of quantifiers in $\phi$. We begin by computing a $(q(\wsn^\HH_c),c)$-\wsm{} to $\HH$, denoted $\vec{X}$, in polynomial time by using Condition~\ref{cond2}. We then proceed by constructing $(G',\vec{X}')$ by Lemma~\ref{lem:similar}. Since each $X'_i\in\vec{X}'$ has size bounded by a constant and $|\vec{X}'|\leq k$, it follows that $\bigcup\vec{X}'$ is a modulator to $\HH$ graphs of cardinality $\bigoh(q(k))$. Then, using Condition~\ref{cond1}, $\MSOMC{\phi}$ admits a kernel of size $p(q(k))$ on $G'$. Finally, since $G$ and $G'$ are $q$-similar, it follows from Fact~\ref{fact:partitiongame} that $G\models \phi$ if and only if $G'\models \phi$.
\end{proof}

Let us briefly discuss the limitations of the above theorem. The condition that $\MSOMC{\phi}$ admits a polynomial (bi)kernel parameterized by $\md^\HH(G)$ is clearly necessary for the rest of the theorem to hold, since $\wsn^\HH(G)\leq \md^\HH(G)$. One might wonder whether a weaker necessary condition could be used instead; specifically, would it be sufficient to require that $\MSOMC{\phi}$ is polynomial-time tractable in $\HH$? This turns out not to be the case, as follows from the following fact.

\begin{fact}[\cite{EibenGanianSzeider15}]
\label{lem:MSOhard}
There exists an MSO sentence $\phi$ and a graph class $\HH$ characterized by a finite set of forbidden induced subgraphs such that $\MSOMC{\phi}$ is polynomial-time tractable on $\HH$ but NP-hard on the class of graphs with $\md^\HH(G)\leq 2$.
\end{fact}

Condition~\ref{cond2} is also necessary for our approach to work, as we need some (approximate) \wsm{}; luckily, Section~\ref{sec:finding} shows that a wide variety of studied graph classes satisfy this condition.
Finally, one can also rule out an extension of Theorem~\ref{thm:main-use} to $\MSOOPT{}{}$ problems (which was possible in the special case considered in Section~\ref{sec:vc}), as we show below.

\lv{\begin{lemma}}
\sv{\begin{lemma}[$\star$]}
\label{lem:MSOOPThard}
There exists an MSO formula $\varphi$ and a graph class $\HH$ characterized by a finite obstruction set such that $\MSOOPT{\leq}{\varphi}$ admits a bikernel parameterized by $\md^\HH$ but is \emph{para}\NP-hard parameterized by $\wsn_1^\HH$.
\end{lemma}

\newcommand{\pfMSOO}[0]{
\begin{proof}
Consider the formula $\varphi(S)=\text{fvs}(S)\vee \text{deg}(S)$, where $\text{fvs}(S)$ expresses that $S$ is a feedback vertex set in $G$ and $\text{deg}(S)$ expresses that $S$ is a modulator to graphs with maximum degree $4$. Let $\HH$ be the class of graphs of maximum degree $4$. 

First, we prove that $\MSOOPT{\leq}{\varphi}$ admits a polynomial bikernel parameterized by $k=\md^\HH(G)$. Consider the following algorithm: given an instance $(G,r)$ of $\MSOOPT{\leq}{\varphi}$, we check if $r\geq k$. If this is the case, we can immediately output YES, since $k$ is the minimum size of a set $S$ satisfying $\text{deg}(S)$. If $r<k$, we compute a polynomial kernel $(G_1,r_1)$ of $\MSOOPT{\leq}{\mtext{\emph{fvs}}}$ from $(G,r)$, and a polynomial kernel $(G_2,r_2)$ for $\MSOOPT{\leq}{\mtext{\emph{deg}}}$ from $(G,r)$. Both $(G_1,r_1)$ and $(G_2,r_2)$ have size bounded by a polynomial of $k$, $G_1$ can be constructed by using any kernelization algorithm for \textsc{Feedback Vertex Set}~\cite{Thomasse10}, and $G_2$ can be constructed by enumerating all obstructions (supergraphs of $K_{1,5}$) and then using a kernelization algorithm for $6$-\textsc{Hitting Set}~\cite{Abu-Khzam10}. Now it is easily observed that $(G,r)$ is a YES instance of $\MSOOPT{\leq}{\varphi}$ if and only if $((G_1, r_1), (G_2, r_2))$ is a YES instance of the following problem: given an instance $(A,a)$ of \textsc{Feedback Vertex Set} and an instance $(B,b)$ of \textsc{$6$-Hitting Set}, answer YES if and only if at least one of $(A,a)$ and $(B,b)$ is a YES instance.

To conclude the proof, we show that $\MSOOPT{\leq}{\varphi}$ is \NP-hard even when $G$ has a $(1,1)$-\wsm{} to $\HH$. We reduce from the (unparameterized) \textsc{Feedback Vertex Set} problem on graphs of degree at most $4$, which is known to be \NP-hard~\cite{GareyJohnson79}. Let $(G,r)$ be an instance of \textsc{Feedback Vertex Set} where $G$ is an $n$-vertex graph of degree at most $4$. We construct the graph $G'$ from $G$ by adding a single vertex $x$ and a disjoint union of $n+1$ stars with five leaves, denoted $S_1,\dots,S_{n+1}$, and then adding an edge from $x$ to the center of each star and from $x$ to a single arbitrary vertex in $G$. Observe that any feedback vertex set in $G$ is also a feedback vertex set in $G'$, and that any set $S$ sayisfying $\text{deg}(S)$ must have cardinality at least $n+1$. Hence $(G,r)$ is a YES-instance of \textsc{Feedback Vertex Set} if and only if $(G',r)$ is a YES-instance of $\MSOOPT{\leq}{\varphi}$. Since $x$ together with the stars added in $G'$ forms a tree (which has rank-width $1$), it follows that this forms a $(1,1)$-\wsm{} to $\HH$.
\end{proof}}
\lv{\pfMSOO}

\sv{
\begin{proof}[Sketch of Proof]
Consider the formula $\varphi(S)=\text{fvs}(S)\vee \text{deg}(S)$, where $\text{fvs}(S)$ expresses that $S$ is a feedback vertex set in $G$ and $\text{deg}(S)$ expresses that $S$ is a modulator to graphs with maximum degree $4$. Let $\HH$ be the class of graphs of maximum degree $4$. 
\end{proof}
}

\subsection{Applications of Theorem~\ref{thm:main-use}}

As our first general application, we consider the results of Gajarsk\' y et al. in~\cite{GajarskyHlinenyObdrzalek13}. Their main result is summarized below.

\begin{fact}[\cite{GajarskyHlinenyObdrzalek13}]\label{fact:kerneltdmod}
Let $\Pi$ be a problem with finite integer index, $\KK$ a class of graphs of bounded expansion,
$d\in \Nat$, and $\HH$ be the class of graphs of treedepth at most $d$. 
Then there exist an algorithm that takes as input $(G,\xi)\in \KK\times \Nat$ and in time $\bigoh(|G|+\log\xi)$ outputs $(G',\xi')$ such that 
\begin{enumerate}
 \item $(G,\xi)\in \Pi$ if and only if $(G',\xi')\in \Pi$; 
 \item $G'$ is an induced subgraph of $G$; and
 \item $|G'|=\bigoh(\md^\HH(G))$.
\end{enumerate}
\end{fact}

The following fact provides a link between the notion of \emph{finite integer index} used in the above result and the $\MSOMC{\varphi}$ problems considered in this paper.

\begin{fact}[\cite{ArnborgEtal93}, see also \cite{BodlaenderEtal09}]\label{fact:MSOisFII}
  For every \MSO{} sentence $\varphi$, it holds that $\MSOMC{\varphi}$
  is finite-state and hence has finite integer index.
\end{fact}

Finally, the following well-known fact is the last ingredient we need to apply our machinery.

\begin{fact}[\cite{NesetrilOssonademendez12}, page 138]\label{fact:tdinducedsubgraphs}
 Let $d\in \Nat$ and $\HH$ be the class of graph of treedepth at most $d$. 
 Then $\HH$ can be characterized by finite set of forbidden induced subgraphs. 
\end{fact}

\lv{\begin{theorem}}
\sv{\begin{theorem}[$\star$]}
\label{thm:kerneltd}
Let $c,d\in \Nat$ and $\HH$ be the class of graphs of treedepth at most $d$. For every \MSO{} sentence $\varphi$, it holds that $\MSOMC{\varphi}$ admits a linear kernel parameterized by $\wsn^{\HH}_c$ on any class of graphs of bounded expansion.
\end{theorem}

\newcommand{\pfkerneltd}[0]{
\begin{proof}
Our proof relies on Theorem~\ref{thm:main-use}.
From Fact~\ref{fact:MSOisFII} and Fact~\ref{fact:kerneltdmod} 
it follows that $\MSOMC{\varphi}$ admits a kernel of size $\bigoh(\md^\HH(G))$. Hence the assumptions of the theorem satisfy Condition~\ref{cond1} of Theorem~\ref{thm:main-use}, where the polynomial $p$ is a linear function.

By Fact~\ref{fact:tdinducedsubgraphs} $\HH$ can be characterized by finite set of forbidden induced subgraphs.
 Therefore, by Theorem~\ref{thm:findwsm} there exists a polynomial algorithm that can find a $(k, c)$-\wsm{} to $\HH$ and $k\le r\cdot\wsn_c^\HH$, where $r\in\Nat$ is a constant (depending on $d$). Therefore, Condition~\ref{cond2} of Theorem~\ref{thm:main-use} is also satisfied, where the polynomial $q$ is also linear function. Hence we conclude that $\MSOMC{\varphi}$ admits a kernel of size $r\cdot\wsn_c^\HH(G)$.
\end{proof}}
\pfkerneltd

As our second general application, we consider \wsm s to the class of forests. 
Lemma~\ref{lem:kernelfvs} shows that feedback vertex set may be used to kernelize any MSO-definable decision problem on graphs of bounded degree. 
\lv{However, before we proceed to the lemma itself, we need to briefly introduce the \emph{protrusion replacement} rule.}

 \newcommand{\protrusions}[0]{ 
 Given a graph $G$, an $r$-\emph{protrusion} is a set of vertices $L\subseteq V(G)$ such that $|N(G-L)|\leq r$ and the treewidth of $G[L]$ is at most $r$. The set $N(G-L)$ is called the \emph{boundary} of $L$. A much more in-depth explanation of protrusion replacement can be found, e.g., in~\cite{BodlaenderEtal09} or \cite{KimLangerPaulReidlRossmanith13}.

\begin{fact}[\cite{BodlaenderEtal09,KimLangerPaulReidlRossmanith13}]
\label{fact:protrusion}
Let $r\in \Nat$ be a constant, $\phi$ be a fixed \MSO{} formula and $G$ be a graph containing an $r$-protrusion $L$. Then there exists a constant-size graph $G'_L$ such that the graph $G'$ obtained by deleting $L$ and making at most $r$ vertices of $G'_L$ adjacent to the neighbors of $L$ satisfies $G'\models \varphi$ iff $G\models \varphi$. Furthermore the graph $G'_L$ may be computed from $G$ and $L$ in polynomial time. 
\end{fact}

The graph $G'_L$ is sometimes called a \emph{representative}.}
\lv{\protrusions}

\lv{\begin{lemma}}
\sv{\begin{lemma}[$\star$]}
\label{lem:kernelfvs}
Let $\FF$ be the class of forests and $d\in \Nat$. For every \MSO{} sentence $\varphi$, it holds that $\MSOMC{\varphi}$ admits a linear kernel parameterized by $\md^{\FF}$ on any class of graphs of degree at most~$d$.
\end{lemma}

\newcommand{\pfkernelfvs}[0]{
\begin{proof}
Let $G'$ be the input graph. We first use Fact~\ref{fact:fvs} to compute a feedback vertex set $X\subseteq V(G)$ of cardinality $k\leq 2\md^\FF(G)$. Let $G'_0$ denote the graph containing all the connected components of $G'-X$ which are not adjacent to $X$. Since $G'_0$ has treewidth $1$ and an empty boundary, by Fact~\ref{fact:protrusion} we can find a constant-size graph $G^*_0$ such that the graph $G^*$ obtained from $G'$ by replacing $G'_0$ with $G^*_0$ satisfies the following: $G^*\models \varphi$ iff $G'\models \varphi$.

Let $G$ be obtained from $G^*$ by deleting $G^*_0$. Then the number of connected components in $H=G-X$ is bounded by $k\cdot d$, since each tree in $H$ is adjacent to at least one vertex in $X$. Let $u$ be the order of the largest representative $G'_L$ as per Fact~\ref{fact:protrusion} for treewidth $1$, boundary-size $2$ and $\varphi$. We proceed by \emph{marking} each vertex in $H$ which is adjacent to some vertex in $X$. This is followed by a secondary marking procedure, where we mark each vertex $v\in V(H)$ of degree at least $3$ with the following property: deleting $v$ separates its connected component in $H$ into at least $3$ connected components of $H$, each containing at least $1$ marked vertex. Observe that the total number of marked vertices is bounded by $2dk$.

Now assume that $H$ contains a tree $T$ of diameter $r$ such that $r>u\cdot (z+1)$, where $z$ is the number of marked vertices in $T$. Let $P$ be a path in $T$ of length $r$. Let us partition $P$ into path segments $P_0,P_1,\dots P_y$, where $P_0$ ends at the first marked vertex on $P$, $P_y$ starts at the first marked vertex on $P$, and each other $P_j$ starts at the $j$-th and ends at the $(j+1)$-th marked vertex on $P$. Then there must exist some path segment of diameter greater than $u$; let us pick an arbitrary such path segment $P_i$ between vertices $a$ and $b$. Let $T'$ be the subtree of $T$ between $a$ and $b$ (including $a,b$); observe that the only marked vertices which may occur in $T'$ are $a$ and $b$. Then $T'$ forms a subgraph in $G$ with a boundary size of $2$, treewidth $1$, and size larger than the representative of $T'$. Hence invoking the replacement procedure on $T'$ results in a new, equivalent graph $G$ of smaller order. The step outlined in this paragraph takes polynomial time and is guaranteed to reduce the order of $G$ by at least $1$.

So, assume that each tree in $H$ has diameter at most $u\cdot (z+1)$, where $z$ is the number of marked vertices in that tree. In particular, this implies that each tree has order at most $ud\cdot (z+1)$. Since the total number of marked vertices in $H$ is at most $dk$, it follows that $V(H)\in \bigoh(k)$. In particular, this implies that $V(G^*)\in \bigoh(k)$, which means that we have a linear kernel.
\end{proof}}
\lv{\pfkernelfvs}

With Lemma~\ref{lem:kernelfvs}, the proof of the following theorem is analogous to the proof of Theorem~\ref{thm:kerneltd}.

\lv{\begin{theorem}}
\sv{\begin{theorem}[$\star$]}
\label{thm:kernelfvs}
Let $c\in \Nat$ and $\FF$ be the class of forests. For every \MSO{} sentence $\varphi$, it holds that $\MSOMC{\varphi}$ admits a linear kernel parameterized by $\wsn^{\FF}_c$ on any class of graphs of bounded degree.
\end{theorem}

\newcommand{\pfthmkernelfvs}[0]{
\begin{proof}
The theorem once again follows from Theorem~\ref{thm:main-use}. In particular, Condition~\ref{cond1} follows directly from Lemma~\ref{lem:kernelfvs} with a linear function $p$. Theorem~\ref{thm:findwsmforest} then gives us a polynomial time algorithm to find a $3$-approximation of $\wsn^{\FF}_c$, satisfying Condition~\ref{cond2} with $q$ also linear. Therefore we conclude by Theorem~\ref{thm:main-use} that $\MSOMC{\varphi}$ admits a linear kernel parameterized by $\wsn^{\HH}_c$ on any class of graphs of bounded degree.
\end{proof}}
\lv{\pfthmkernelfvs}




\section{Conclusion}
\label{sec:hardness}
Our results show that measuring the structure of modulators can lead to an interesting and, as of yet, relatively unexplored spectrum of structural parameters. Such parameters have the potential of combining the best of decomposition-based techniques and modulator-based techniques, and can be applied both in the context of kernelization (as demonstrated in this work) and FPT algorithms~\cite{EibenGanianSzeider15}. We believe that further work in the direction of modulators will allow us to push the frontiers of tractability towards new, uncharted classes of inputs.

One possible direction for future research is the question of whether the class of $\MSO$-definable problems considered in Theorem~\ref{thm:main-use} can be extended to other finite-state problems. It would of course also be interesting to see more applications of Theorem~\ref{thm:main-use} and new methods for approximating \wsm s. Last but not least, we mention that the split-modules used in the definition of our parameters could in principle be refined to less restrictive notions (for instance cuts of constant \emph{cut-rank}~\cite{OumSeymour06}); such a relaxed parameter could still be used to obtain polynomial kernels, as long as there is a way of efficiently approximating or computing such modulators.

\bibliographystyle{abbrv} \bibliography{literature}

\newpage

\sv{
\appendix

\section{Omissions for Subsection~\ref{sub:mso}}
\introtypes

\section{Omissions for Section~\ref{sec:kcwsm}}
\sparsity

\section{Omissions and Proofs for Section~\ref{sec:vc}}
\subsection{Proof of Proposition~\ref{prop:betterVC}}
\pfbetterVC

\subsection{Omissions Required for the Proof of Theorem~\ref{thm:msovc}}
\MSOVC

\subsection{Proof of Theorem~\ref{thm:msovc}}
\pfmsovc

\subsection{Omissions Required for the Proof of Theorem~\ref{thm:ann}}
\annotation

\subsection{Proof of Theorem~\ref{thm:ann}}
\pfann

\section{Proofs for Section~\ref{sec:finding}}
\subsection{Proof of Lemma~\ref{lem:longcycle}}
\pflongcycle

\subsection{Proof of Theorem~\ref{thm:findwsmforest}}
\pffindwsmforest

\subsection{Proof of Lemma~\ref{lem:hitting}}
\pflemhitting

\section{Omissions and Proofs for Section~\ref{sec:using}}
\subsection{Proof of Lemma~\ref{lem:MSOOPThard}}
\pfMSOO

\subsection{Proof of Theorem~\ref{thm:kerneltd}}
\pfkerneltd

\subsection{Omissions Required for the Proof of Lemma~\ref{lem:kernelfvs}}
\protrusions

\subsection{Proof of Lemma~\ref{lem:kernelfvs}}
\pfkernelfvs

\subsection{Proof of Theorem~\ref{thm:kernelfvs}}
\pfthmkernelfvs

}

\end{document}